\documentclass{article}[12pt]
\usepackage{amsmath,amsthm,amssymb,amsfonts,amscd,eucal}

\numberwithin{equation}{section}

\def\ca{{\mathcal A}}

\def\cf{{\mathcal F}}

\def\cam{{\mathcal M}}

\def\bc{{\mathbb C}}

\def\bn{{\mathbb N}}
\def\br{{\mathbb R}}

\def\a{\alpha}
\def\b{\beta}

\def\l{\lambda}       
\def\m{\mu}

\def\r{\rho}
\def\s{\sigma}
    
\def\t{\tau}

\def\o{\omega}         \def\O{\Omega}

\def\itm#1{\item{$(#1)$}}

\def\ov{\overline}

\DeclareMathOperator{\re}{Re}

\DeclareMathOperator{\Var}{Var} 
\DeclareMathOperator{\Cov}{Cov}
\DeclareMathOperator{\Corr}{Corr} 
\DeclareMathOperator{\Tr}{Tr}

\newcommand{\set}[1]{\left\{#1\right\}}

\newtheorem{Thm}{Theorem}[section]

\newtheorem{Prop}[Thm]{Proposition}
\newtheorem{Lemma}[Thm]{Lemma}

\theoremstyle{definition}
\newtheorem{Dfn}[Thm]{Definition}

\theoremstyle{remark}
\newtheorem{Rem}[Thm]{Remark}

\begin{document}

\title{Uncertainty principle for Wigner-Yanase-Dyson information in
semifinite von Neumann algebras}

\author{ Paolo Gibilisco\footnote{Dipartimento SEFEMEQ, Facolt\`a di
Economia, Universit\`a di Roma ``Tor Vergata", Via Columbia 2, 00133
Rome, Italy.  Email: gibilisco@volterra.uniroma2.it -- URL:
http://www.economia.uniroma2.it/sefemeq/professori/gibilisco}
 \ and Tommaso Isola\footnote{Dipartimento
di Matematica, Universit\`a di Roma ``Tor Vergata", Via della Ricerca
Scientifica, 00133 Rome, Italy.  Email: isola@mat.uniroma2.it -- URL:
http://www.mat.uniroma2.it/$\sim$isola} }

\maketitle

\begin{abstract}
    In \cite{Ko} Kosaki proved an uncertainty principle for matrices,
    related to Wigner-Yanase-Dyson information, and asked if a similar
    inequality could be proved in the von Neumann algebra setting.  In
    this paper we prove such an uncertainty principle in the
    semifinite case.

\smallskip

\noindent 2000 {\sl Mathematics Subject Classification.} Primary
62B10, 94A17; Secondary 46L30, 46L60.

\noindent {\sl Key words and phrases.} Uncertainty principle,
Wigner-Yanase-Dyson information.
\end{abstract}

\section{Introduction}

Let $M_n:=M_n(\mathbb{C})$ (resp.$M_{n,sa}:=M_n(\mathbb{C})_{sa}$) be
the set of all $n \times n$ complex matrices (resp.  all $n \times n$
self-adjoint matrices).  Let ${\cal D}_n^1$ be the set of strictly
positive density matrices namely
$$
{\cal D}_n^1=\{\rho \in M_n : {\rm Tr} \rho=1, \, \rho>0 \}.
$$

\begin{Dfn}
    For $A,B \in M_{n,sa}$ and $\rho \in {\cal D}_n^1$ define
    covariance and variance as
    \begin{align*}
	\Cov_{\rho}(A,B) & :={\rm Tr}(\rho A B)-{\rm Tr}(\rho
	A)\cdot{\rm Tr}(\rho B) \\
	\Var_{\rho}(A) & :={\rm Tr}(\rho A^2)-{\rm Tr}(\rho A)^2.
    \end{align*}
\end{Dfn}

Then the well known Schr{\"o}dinger and Heisenberg uncertainty
principles are given in the following 

\begin{Thm} {\rm \cite{Heis,Schr}}

For $A,B \in M_{n,sa}$ and $\rho \in {\cal
D}_n^1$ one has
$$
{\rm Var}_{\rho}(A){\rm Var}_{\rho}(B)-|\re \Cov_{\rho}(A,B)|^2
\geq
\frac{1}{4}\vert {\rm Tr}(\rho[A,B])\vert^2,
$$
that implies
$$
{\rm Var}_{\rho}(A){\rm Var}_{\rho}(B)
\geq
\frac{1}{4}\vert {\rm Tr}(\rho[A,B])\vert^2.
$$
\end{Thm}

Recently a different uncertainty principle has been found 
\cite{LZ,LQa,LQb,Ko,YFK}. 

\begin{Dfn}
    For $A,B \in M_{n,sa}$, $\b\in(0,1)$, and $\rho \in {\cal D}_n^1$
    define $\b$-correlation and $\b$-information as
    \begin{align*}
	\Corr_{\rho,\b}(A,B) & :={\rm Tr}(\rho A B)-{\rm Tr}(\rho^{\b}
	A\r^{1-\b}B) \\
	I_{\rho,\b}(A) & := \Corr_{\r,\b}(A,A) \equiv Tr(\r A^{2}) -
	\Tr(\r^{\b}A\r^{1-\b}A).
	    \end{align*}
    The latter coincides with the Wigner-Yanase-Dyson information.
\end{Dfn}

\begin{Thm} \label{Thm:KoIneq}
    $$
    \Var_{\r}(A)\Var_{\r}(B) - |\re\Cov_{\r}(A,B)|^{2} \geq 
    I_{\r,\b}(A)I_{\r,\b}(B) - |\re\Corr_{\r,\b}(A,B)|^{2}.
    $$
\end{Thm}

Kosaki \cite{Ko} asked if the previous inequality, which makes 
perfect sense in a von Neumann algebra setting, could indeed be
proved.  In the sequel, we provide such a proof in the
semifinite case.

In closing, we mention that different generalizations of Theorem 
\ref{Thm:KoIneq} have been recently obtained by the authors 
\cite{GiIs07a,GiIs07b,GII01,GII02,GII03,GII04}.

\section{Auxiliary lemmas}
\label{sec:auxiliary}

In all this Section we let $(\cam,\t)$ be a semifinite von Neumann
algebra with a n.s.f. trace, and denote by $Proj(\cam)$ the set of
orthogonal projections in $\cam$, and by $\ov{\cam}$ the topological
$^{*}$-algebra of $\t$-measurable operators.  We fix $\r,\s\in
\ov{\cam}_{sa}$, with spectral decompositions
$\r=\int_{-\infty}^{+\infty} \l\, d e_{\r}(\l)$, and
$\s=\int_{-\infty}^{+\infty} \l\, d e_{\s}(\l)$. 

Finally, we denote by $\ca$ the algebra generated by the sets
$\O_{1}\times \O_{2}$, for $\O_{1},\O_{2}$ Borel subsets of $\br$, and
observe that $\s(\ca)$, the $\s$-algebra generated by $\ca$, coincides
with the Borel subsets of $\br^{2}$.

\begin{Lemma} \label{Lem:Extension}
    Let $a,b\in \cam \cap L^{2}(\cam,\t)$.  Let
    $\m_{ab}(\O_{1}\times\O_{2}) :=
    \t(e_{\r}(\O_{1})a^{*}e_{\s}(\O_{2})b)$, for $\O_{1},\O_{2}$ Borel
    subsets of $\br$.  Then $\m_{ab}$ extends uniquely to a bounded
    Borel measure on $\br^{2}$.
\end{Lemma}
\begin{proof}
    For $\O\subset\br$ Borel subset, $x\in L^{2}(\cam,\t)$, let
    $P(\O)x:=e_{\r}(\O)x$, $Q(\O)x:=xe_{\s}(\O)$.  Then, $P,Q$ are
    commuting Borel spectral measures on $L^{2}(\cam,\t)$, and their
    product $P\otimes Q (\O_{1}\times \O_{2}):= P(\O_{1})Q(\O_{2})$
    extends uniquely to a Borel spectral measure on $\br^{2}$
    (\cite{BS}, Chapter 5).  Observe that $\m_{ab}(\O_{1}\times\O_{2})
    = \t(P\otimes Q(\O_{1}\times \O_{2})(a^{*})\cdot b)$, and, if
    $\set{A_{n}}$ is a sequence of disjoint Borel sets, then $P\otimes
    Q(\cup A_{n})(a^{*}) = \sum_{n} P\otimes Q(A_{n})(a^{*})$
    converges in $L^{2}(\cam,\t)$, so that $\t(P\otimes Q(\cup
    A_{n})(a^{*})\cdot b)$ is well defined.  So $\m_{ab} = \t(P\otimes
    Q(\cdot)(a^{*})\cdot b)$ is the desired extension.  
    
    Observe now that $\m_{ab}$ is a bounded Borel (complex) measure 
    on $\ca$. Indeed, with $A\in\ca$,
    $$
    |\m_{ab}(A)|^{2} = |\t(P\otimes Q(A)(a^{*})\cdot b)|^{2} \leq 
    \|P\otimes Q (A)(a^{*})\|_{L^{2}} \|b\|_{L^{2}} \leq 
    \|a\|_{L^{2}} \|b\|_{L^{2}}.
    $$
    Therefore, by \cite{Rao} Corollary 4.4.6, there is a unique
    extension of $\m_{ab}$ to a bounded (complex) measure on $\s(\ca)$,
    the $\s$-algebra generated by $\ca$, $i.e.$ the Borel subsets of
    $\br^{2}$.
\end{proof}

\begin{Lemma} \label{Lem:Properties}
    Let $a,b\in \cam\cap L^{2}(\cam,\t)$.  Then 
    
    \itm{i} $\m_{ab} = \frac14 \sum_{k=1}^{4} (-i)^{k} 
    \m_{a+i^{k}b,a+i^{k}b}$,
    
    \itm{ii} if $\s=\r$, $\m_{aa}$ is a real positive  measure,
    
    \itm{iii} if $a,b$ are self-adjoint, $\re \m_{ab} = \re \m_{ba}$.
\end{Lemma}
\begin{proof}
    $(i)$ is standard.
    
    $(ii)$ Let $\O_{1},\O_{2}$ be Borel sets in $\br$, and set 
    $e_{j}:=e_{\r}(\O_{j})$, $j=1,2$. Then 
    $\m_{aa}(\O_{1}\times\O_{2}) = \t(e_{1}a^{*}e_{2}a) = 
    \t((e_{2}ae_{1})^{*}e_{2}ae_{1})\geq 0$, and the thesis follows 
    by uniqueness of the extension from $\ca$ to $\s(\ca)$.
    
    $(iii)$ Let $\O_{1},\O_{2}$ be Borel sets in $\br$, and set
    $e_{1}:=e_{\r}(\O_{1})$, $e_{2}:=e_{\s}(\O_{2})$.  Then $\re
    \m_{ab}(\O_{1}\times \O_{2}) = \re \t(e_{1}ae_{2}b) = \re
    \t(be_{2}ae_{1}) = \re \t(e_{1}be_{2}a) = \re \m_{ba}(\O_{1}\times
    \O_{2})$.
\end{proof}

\begin{Lemma} \label{Lem:FunctionalCalculus}
    Let $a,b\in\cam \cap L^{2}(\cam,\t)$.   Let 
    $g,h:\br\to\bc$ be bounded Borel functions. Then
    $$
    \t(g(\r)a^{*}h(\s)b) = \iint g(x)h(y)\, d\m_{ab}(x,y).
    $$
\end{Lemma}
\begin{proof}
    We use notation as in the proof of Lemma \ref{Lem:Extension}.  Let
    $s=\sum_{i=1}^{h} s_{i}\chi_{A_{i}}$, $t=\sum_{j=1}^{k}
    t_{j}\chi_{B_{j}}$ be simple Borel functions.  Then
    \begin{align*}
	\t(s(\r)a^{*}t(\s)b) & = \sum_{i=1}^{h}\sum_{j=1}^{k} s_{i}t_{j}
	\t(\chi_{A_{i}}(\r)a^{*}\chi_{B_{j}}(\s)b) =
	\sum_{i=1}^{h}\sum_{j=1}^{k} s_{i}t_{j} \t(P\otimes
	Q(A_{i}\times B_{j})(a^{*})\cdot b) \\
	& = \sum_{i=1}^{h}\sum_{j=1}^{k} s_{i}t_{j} \iint
	\chi_{A_{i}\times B_{j}}\, d\m_{ab} = \iint s(x)t(y)\,
	d\m_{ab}(x,y).
    \end{align*}
    Let now $g,h$ be bounded Borel functions, and
    $\set{s_{m}},\set{t_{n}}$ sequences of simple Borel functions such
    that $s_{m}\to g$, $t_{n}\to h$ and $|s_{m}|\leq |g|$,
    $|t_{n}|\leq |h|$.  Denote $r_{n}(x,y):= s_{n}(x)t_{n}(y)$,
    $k(x,y):=g(x)h(y)$.  Then, by (\cite{BS}, Theorem V.3.2),
    $s_{n}(\r)a^{*}t_{n}(\s) = P\otimes Q(r_{n})(a^{*}) \to P\otimes
    Q(k)(a^{*}) = g(\r)a^{*}h(\s)$ in $L^{2}(\cam,\t)$, so that
    $\t(s_{n}(\r)a^{*}t_{n}(\s)b) \to \t(g(\r)a^{*}h(\s)b)$. 
    Moreover, $\iint r_{n}\, d\m_{ab} \to \iint k \, d\m_{ab}$,
    because $\m_{ab}$ is a bounded measure.  The thesis follows.
\end{proof}

\begin{Lemma} \label{Lem:DoubleInt}
    Let $a,b\in\cam \cap L^{2}(\cam,\t)$, $\r\in
    L^{1}(\cam,\t)_{+}$, $\b\in(0,1)$.  Then
    $$
    \t(\r^{\b}a^{*}\r^{1-\b}b) = \iint_{[0,\infty)^{2}} x^{\b}y^{1-\b}
    \, d\m_{ab}(x,y).
    $$
\end{Lemma}
\begin{proof}
    Let $n\in\bn$, and set 
    $$
    f_{n}(x) := 
    \begin{cases}
	x, & 0 \leq x \leq n\\
	0, & \text{ else}
    \end{cases}
    \qquad
    f(x) := 
    \begin{cases}
	x, & x \geq 0\\
	0, & x<0.
    \end{cases}    
    $$
    Then
    $$
    \t(f_{n}(\r)^{\b}a^{*}f_{n}(\r)^{1-\b}b) = \int_{\br^{2}} 
    f_{n}(x)^{\b}f_{n}(y)^{1-\b} \, d \m_{ab}(x,y).
    $$
    Observe now that $f_{n}(\r)^{\b} \to f(\r)^{\b} = \r^{\b}$ in
    $L^{1/\b}(\cam,\t)$, so that $f_{n}(\r)^{\b}a^{*}f_{n}(\r)^{1-\b}b
    \to \r^{\b}a^{*}\r^{1-\b}b$ in $L^{1}(\cam,\t)$, which implies
    $$
    \t(f_{n}(\r)^{\b}a^{*}f_{n}(\r)^{1-\b}b) \to
    \t(\r^{\b}a^{*}\r^{1-\b}b).
    $$    
    Moreover, in case $\s=\r$, $\m_{aa}$ is a positive measure, so
    that, by monotone convergence, 
    $$
    \int_{\br^{2}} f_{n}(x)^{\b}f_{n}(y)^{1-\b} \, d \m_{aa}(x,y) \to
    \iint_{[0,\infty)^{2}} x^{\b}y^{1-\b} \, d\m_{aa}(x,y).
    $$    
    Therefore, the thesis holds for $a=b$. By polarization (Lemma 
    \ref{Lem:Properties} $(i)$) the result is true in general.
\end{proof}

\begin{Lemma} \label{Lem:PositiveMeasure}
    Let $a,b\in\cam \cap L^{2}(\cam,\t)$.  Then,
    $$
    \m := \m_{aa} \otimes \m_{bb} + \m_{bb} \otimes \m_{aa} -2 \re
    \m_{ab} \otimes \re \m_{ab}
    $$
    is a real positive Borel measure on $\br^{4}$. 
\end{Lemma}
\begin{proof}
    Indeed, if $\O_{1},\ldots,\O_{4}\subset \br$ are measurable
    subsets, and $E_{j}:=e_{\r}(\O_{j})\in Proj(\cam)$, $j=1,3$, 
    $E_{j}:=e_{\s}(\O_{j})\in Proj(\cam)$, $j=2,4$,
    then
    \begin{align*}
	\m(\O_{1}\times\cdots\times\O_{4}) & =
	\t(E_{1}a^{*}E_{2}a) \cdot \t(E_{3}b^{*}E_{4}b) +
	\t(E_{3}a^{*}E_{4}a) \cdot \t(E_{1}b^{*}E_{2}b) \\
	& \quad - 2 \re \t(E_{1}a^{*}E_{2}b) \cdot \re
	\t(E_{3}a^{*}E_{4}b) \\
	& \geq \t(E_{1}a^{*}E_{2}a) \cdot \t(E_{3}b^{*}E_{4}b)
	+ \t(E_{3}a^{*}E_{4}a) \cdot \t(E_{1}b^{*}E_{2}b) \\
	& \quad - 2 | \t(E_{1}a^{*}E_{2}b) | \cdot
	|\t(E_{3}a^{*}E_{4}b)|.
    \end{align*}
    Moreover, 
    \begin{align*}
	| \t(E_{1}a^{*}E_{2}b) | & = | \t \bigl(
	(E_{2}aE_{1})^{*} E_{2}bE_{1} \bigr) | \\
	& \leq \t \bigl( (E_{2}aE_{1})^{*} E_{2}aE_{1}
	\bigr)^{1/2} \t \bigl( (E_{2}bE_{1})^{*} E_{2}bE_{1}
	\bigr)^{1/2} \\
	& = \t(E_{1}a^{*}E_{2}a)^{1/2} \cdot
	\t(E_{1}b^{*}E_{2}b)^{1/2}.
    \end{align*}
    Therefore, setting $\a_{1}:= \t(E_{1}a^{*}E_{2}a)^{1/2}$, $\b_{1}
    := \t(E_{1}b^{*}E_{2}b)^{1/2}$, $\a_{2} :=
    \t(E_{3}a^{*}E_{4}a)^{1/2}$, $\b_{2} :=
    \t(E_{3}b^{*}E_{4}b)^{1/2}$, we have
    $\m(\O_{1}\times\cdots\times\O_{4}) \geq \a_{1}^{2}\b_{2}^{2} +
    \a_{2}^{2}\b_{1}^{2}-2\a_{1}\b_{1}\a_{2}\b_{2} \geq 0$, and the
    thesis follows by standard measure theoretic arguments.
\end{proof}

\section{The main result} 
\label{sec:main}

Let $(\cam,\t)$ be a semifinite von Neumann algebra with a n.s.f.
trace.  Let $\o$ be a normal state on $\cam$, and $\r_{\o}\in
L^{1}(\cam,\t)_{+}$ be such that $\o(x)=\t(\r_{\o}x)$, for $x\in\cam$.
Then, for any $A,B\in\cam_{sa}$, $\b\in(0,1)$, we set

\begin{Dfn}
    \begin{align*}	    
	\Cov_{\o}(A,B) & := \o(AB)-\o(A)\o(B) \equiv \t(\r_\o
	AB)-\t(\r_\o A)\t(\r_\o B), \\
	\Var_{\o}(A) & := \Cov_{\o}(A,A) \equiv \o(A^{2})-\o(A)^{2}
	\equiv \t(\r_\o A^2)-\t(\r_\o A)^2,\\
	\Corr_{\o,\b}(A,B) & := \t(\r_\o AB) - \t(\r_\o^\b
	A\r_\o^{1-\b} B), \\
	I_{\o,\b}(A) & := \Corr_{\o,\b}(A,A) \equiv \t(\r_\o A^2) -
	\t(\r_\o^\b A\r_\o^{1-\b} A).
    \end{align*}
\end{Dfn}

\begin{Prop}
    Let $A_{0} := A-\o(A)I$, $B_{0}:= B-\o(B)I$. Then
    \begin{align*}
	\Cov_{\o}(A,B) & = \t(\r_\o A_{0}B_{0}),\\
	\Corr_{\o,\b}(A,B) & = \t(\r_\o A_{0}B_{0}) - \t(\r_\o^\b A_0
	\r_\o^{1-\b} B_0).
    \end{align*}
\end{Prop}

\begin{Thm} \label{Thm:main}
   For any $A,B\in\cam_{sa}$, $\b\in(0,1)$, we have
   $$
   \Var_{\o}(A)\Var_{\o}(B)-|\re
    \Cov_{\o}(A,B)|^{2} \geq I_{\o,\b}(A)I_{\o,\b}(B) -
    |\re\Corr_{\o,\b}(A,B)|^{2}.
   $$
\end{Thm}
\begin{proof}
    To start with, let us assume that $A,B\in \cam\cap 
    L^{2}(\cam,\t)$. Set 
    \begin{align*}
	\cf & := \Var_{\o}(A)\Var_{\o}(B)-|\re \Cov_{\o}(A,B)|^{2} -
	I_{\o,\b}(A)I_{\o,\b}(B) + |\re\Corr_{\o,\b}(A,B)|^{2} \\
	& = \t(\r_{\o}A_{0}^{2}) \cdot
	\t(\r_{\o}^{\b}B_{0}\r_{\o}^{1-\b}B_{0}) + \t(\r_{\o}B_{0}^{2})
	\cdot \t(\r_{\o}^{\b}A_{0}\r_{\o}^{1-\b}A_{0}) -
	\t(\r_{\o}^{\b}A_{0}\r_{\o}^{1-\b}A_{0}) \cdot
	\t(\r_{\o}^{\b}B_{0}\r_{\o}^{1-\b}B_{0}) \\
	&\quad -2 \re \t(\r_{\o}A_{0}B_{0}) \cdot \re
	\t(\r_{\o}^{\b}A_{0}\r_{\o}^{1-\b}B_{0}) + \bigl( \re
	\t(\r_{\o}^{\b}A_{0}\r_{\o}^{1-\b}B_{0}) \bigr)^{2}.
    \end{align*}
    Then, using Lemma \ref{Lem:DoubleInt} and symmetries of the
    integrands, we obtain
    \begin{align*}
	\cf_{1} & := \t(\r_{\o}A_{0}^{2}) \cdot
	\t(\r_{\o}^{\b}B_{0}\r_{\o}^{1-\b}B_{0}) + \t(\r_{\o}B_{0}^{2})
	\cdot \t(\r_{\o}^{\b}A_{0}\r_{\o}^{1-\b}A_{0}) -
	\t(\r_{\o}^{\b}A_{0}\r_{\o}^{1-\b}A_{0}) \cdot
	\t(\r_{\o}^{\b}B_{0}\r_{\o}^{1-\b}B_{0}) \\
	& = \int_{[0,\infty)^{4}} \l_{1}\l_{3}^{\b}\l_{4}^{1-\b}\, d
	\m_{A_{0}A_{0}}\otimes \m_{B_{0}B_{0}} (\l_{1},\ldots,\l_{4})
	+ \int_{[0,\infty)^{4}} \l_{3}\l_{1}^{\b}\l_{2}^{1-\b}\, d
	\m_{A_{0}A_{0}}\otimes \m_{B_{0}B_{0}} (\l_{1},\ldots,\l_{4})
	\\
	& \quad - \int_{[0,\infty)^{4}}
	\l_{1}^{\b}\l_{2}^{1-\b}\l_{3}^{\b}\l_{4}^{1-\b}\, d
	\m_{A_{0}A_{0}}\otimes \m_{B_{0}B_{0}} (\l_{1},\ldots,\l_{4})
	\\
	& = \frac12 \int_{[0,\infty)^{4}} \Bigl(
	(\l_{1}+\l_{2})\l_{3}^{\b}\l_{4}^{1-\b} +
	\l_{1}^{\b}\l_{2}^{1-\b}(\l_{3}+\l_{4}) -2
	\l_{1}^{\b}\l_{2}^{1-\b}\l_{3}^{\b}\l_{4}^{1-\b} \Bigr) d
	\m_{A_{0}A_{0}}\otimes \m_{B_{0}B_{0}} (\l_{1},\ldots,\l_{4}),
    \end{align*}
    \begin{align*}
	\cf_{2} & := 2 \re \t(\r_{\o}A_{0}B_{0}) \cdot \re
	\t(\r_{\o}^{\b}A_{0}\r_{\o}^{1-\b}B_{0}) - \bigl( \re
	\t(\r_{\o}^{\b}A_{0}\r_{\o}^{1-\b}B_{0}) \bigr)^{2} \\
	& = 2\int_{[0,\infty)^{4}} \l_{1}\l_{3}^{\b}\l_{4}^{1-\b}\, d
	\re \m_{A_{0}B_{0}} \otimes \re \m_{A_{0}B_{0}}
	(\l_{1},\ldots,\l_{4}) \\
	& \quad - \int_{[0,\infty)^{4}}
	\l_{1}^{\b}\l_{2}^{1-\b}\l_{3}^{\b}\l_{4}^{1-\b}\, d \re
	\m_{A_{0}B_{0}} \otimes \re \m_{A_{0}B_{0}}
	(\l_{1},\ldots,\l_{4}) \\
	& = \frac12 \int_{[0,\infty)^{4}} \Bigl(
	(\l_{1}+\l_{2})\l_{3}^{\b}\l_{4}^{1-\b} +
	\l_{1}^{\b}\l_{2}^{1-\b}(\l_{3}+\l_{4}) -2
	\l_{1}^{\b}\l_{2}^{1-\b}\l_{3}^{\b}\l_{4}^{1-\b} \Bigr)\, d
	\re \m_{A_{0}B_{0}} \otimes \re \m_{A_{0}B_{0}}
	(\l_{1},\ldots,\l_{4}).
    \end{align*}
    So that, using the notation of Lemma \ref{Lem:PositiveMeasure},
    $$
    \cf = \cf_{1}-\cf_{2} = \frac14 \int_{[0,\infty)^{4}} \Bigl(
    (\l_{1}+\l_{2})\l_{3}^{\b}\l_{4}^{1-\b} +
    \l_{1}^{\b}\l_{2}^{1-\b}(\l_{3}+\l_{4}) -2
    \l_{1}^{\b}\l_{2}^{1-\b}\l_{3}^{\b}\l_{4}^{1-\b} \Bigr) \,
    d\m(\l_{1},\ldots,\l_{4}).  
    $$ 
    Since $\m$ is a real positive measure on $[0,\infty)^{4}$, because
    of Lemma \ref{Lem:PositiveMeasure}, and
    \begin{align*}
	& (\l_{1}+\l_{2})\l_{3}^{\b}\l_{4}^{1-\b} +
	\l_{1}^{\b}\l_{2}^{1-\b}(\l_{3}+\l_{4}) -2
	\l_{1}^{\b}\l_{2}^{1-\b}\l_{3}^{\b}\l_{4}^{1-\b} \\
	& = (\l_{1}+\l_{2}
	-\l_{1}^{\b}\l_{2}^{1-\b})\l_{3}^{\b}\l_{4}^{1-\b} +
	\l_{1}^{\b}\l_{2}^{1-\b}(\l_{3}+\l_{4}-\l_{3}^{\b}\l_{4}^{1-\b})
	\ge 0,
    \end{align*}
    we get $\cf\geq 0$, which is what we wanted to prove.
    
    Finally, to extend the validity of the inequality from
    $\cam_{sa}\cap L^{2}(\cam,\t)$ to $\cam_{sa}$, let us observe that
    $\cam_{sa}\cap L^{2}(\cam,\t)$ is $\s$-weakly dense in
    $\cam_{sa}$, and $a\in\cam\mapsto \t(\r_{\o}ab)$, $b\in\cam\mapsto
    \t(\r_{\o}ab)$, $a\in\cam\mapsto \t(\r^{\b}a\r^{1-\b}b)$, and
    $b\in\cam\mapsto \t(\r^{\b}a\r^{1-\b}b)$ are $\s$-weakly
    continuous.
\end{proof}

\begin{Rem}
    Observe that, reasoning as in \cite{Ko} Theorem 5, one can prove
    that the function 
    $$
    g(\b) := \Var_{\o}(A)\Var_{\o}(B)-|\re \Cov_{\o}(A,B)|^{2} -
    I_{\o,\b}(A)I_{\o,\b}(B) + |\re\Corr_{\o,\b}(A,B)|^{2}
    $$ 
    is monotone increasing on the interval $[\frac12,1)$.  Therefore,
    the best bound in Theorem \ref{Thm:main} is given by $\b=\frac12$,
    $i.e.$ by the Wigner-Yanase information.
\end{Rem}

\end{document}